\documentclass[a4paper,twoside]{article}
\usepackage{epsfig}
\usepackage{subfigure}
\usepackage{calc}
\usepackage{amssymb}
\usepackage{amstext}
\usepackage{amsmath}
\usepackage{amsthm}
\usepackage{multicol}
\usepackage{pslatex}
\usepackage{apalike}
\usepackage{url}
\usepackage{hyperref}  
\usepackage{SCITEPRESS}     % Please add other packages that you may need BEFORE the SCITEPRESS.sty package. 

\begin{document}

\newtheorem{lemma}{Lemma}{\bf}{\it}
\newtheorem{theorem}{Theorem}{\bf}{\it}
%\makeatletter
%\g@addto@macro{\UrlBreaks}{\UrlOrds}
%\makeatother

\title{A New Approach to GraphMaps, a System Browsing\\ Large Graphs as Interactive Maps}%\thanks{Accepted at the International Conference on Information Visualization Theory and Applications (IVAPP),  Madeira, Portugal, 2018.}}

\author{\authorname{Debajyoti Mondal\sup{1}  and Lev Nachmanson\sup{2}}% and Third Author Name\sup{2}}
\affiliation{\sup{1}Department of Computer Science, University of Saskatchewan, Saskatoon, SK, Canada}
\affiliation{\sup{2}Microsoft Research, Redmond, WA,  U.S.A.}
\email{dmondal@cs.usask.ca, levnach@microsoft.com}
}

%\keywords{The paper must have at least one keyword. The text must be set to 9-point font size and without the use of bold or italic font style. For more than one keyword, please use a comma as a separator. Keywords must be titlecased.}

%\abstract{The abstract should summarize the contents of the paper and should contain at least 70 and at most 200 words. The text must be set to 9-point font size.}

\abstract{A GraphMaps is a system that visualizes a graph using zoom levels, which is similar to a geographic map visualization. GraphMaps reveals the structural properties of the graph and enables users to explore the graph in a natural way by using the standard zoom and pan operations. The available implementation of GraphMaps faces many challenges such  as  the number of zoom levels may be large, nodes may be unevenly distributed to different levels,   shared edges may create ambiguity due to the selection of multiple  nodes. In this paper, we develop an algorithmic framework to construct GraphMaps from any given mesh (generated from a 2D point set), and for any given number of zoom levels. We demonstrate our approach introducing competition mesh, which is simple to construct, has a low dilation and high angular resolution. We present an algorithm for assigning nodes to zoom levels that minimizes the change in the number  of nodes on visible on the screen while the user zooms in and out between the levels. We think that keeping this change small facilitates smooth browsing of the graph. We also propose  new node selection techniques to cope with some of the challenges of the GraphMaps approach.}
%\begin{minipage}{\linewidth} 
%\centering
%\smallskip 
%{\includegraphics[width=.3\textwidth, height = 3cm]{figures/unread}}
%\hfill
%{\includegraphics[width=.3\textwidth, height = 3cm]{figures/flight2}}
%\hfill
%{\includegraphics[width=.3\textwidth, height = 3cm]{figures/flight3}}
%\captionof{figure}{A partial display of a flight network with approximately 3K nodes and 19K edges. (left) Traditional node-link diagram over the airports of North America. (middle) The  top-level of a GraphMaps visualization based on our approach, where the airports TUS and YWG are selected. (right) A view after zoom in.}
%\label{abs}
%\end{minipage}
%}

\onecolumn \maketitle \normalsize \vfill

\section{\uppercase{Introduction}}
 
Traditional data visualization systems render all the vertices and edges of the graph on a single screen. For large graphs, this approach requires rendering many objects on the screen, which overwhelms the user. A GraphMaps %system overcomes this problem by visualizing the graph using zoom levels, which is similar to a geographic map visualization. We first review the major features of GraphMaps, and give some definitions. The
 system confronts the challenge of visualizing large graphs by enabling the users to browse the graphs as interactive maps. Like Google or Bing Maps, a GraphMaps system visualizes the high priority features on the top level, and as we zoom in, the low priority entities start to appear in the subsequent levels. %, the same way as less significant details appear on the digital maps when zooming in. 
To achieve this effect,  for a given graph $G$ and a positive integer $k > 0$,  GraphMaps creates the graphs $G_1, G_2,\ldots, G_k$, where $G_i$, $1 \le i < k$, is an induced subgraph of $G_{i+1}$, and $G_k = G$. 

\begin{figure*}[pt]
  \centering
  \includegraphics[width=.7\textwidth]{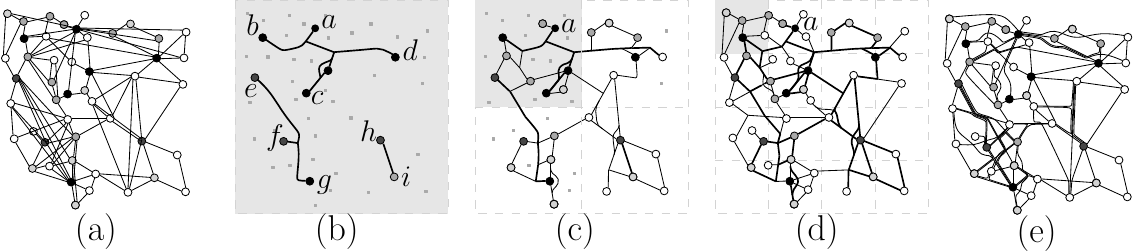}  
  \caption{(a)  A %traditional
 node-link diagram of a graph $G$. (b--d) A GraphMaps visualization of $G$. (e) An example of edge bundling.
 %(f--h) Node selections and zoom in.%Illustration for node selections and zoom in. 
}
\label{fig:introduction}
\end{figure*}

The graph $G_i$, where $ 1\le i\le k$, corresponds to the $i$th zoom level. Assume that the nodes of $G$ are ranked by their importance. The discussion on what a node importance is and how the ranking is obtained, is out of the scope of the paper, but by default GraphMaps uses Pagerank~\cite{bp12} to obtain such a ranking.  Let $V(G_i)$ be the nodes of $G_i$. We build graphs $G_i$ in such a way that the nodes of $G_i$ are equally or more important than the nodes of $V(G_{i+1}) \setminus V(G_i)$. At the top view, we render the graph $G_1$. As we zoom in and the zoom reaches $2^{i-1}$,  the rendering switches from $G_{i-1}$ to $G_i$, exposing less important nodes and their incident edges. To create spatial stability GraphMaps keeps the node positions fixed, and the rendering of edges changes incrementally between $G_i$ and $G_{i+1}$ as described in Section~\ref{sec:cm}.

By browsing a graph with GraphMaps, the user obtains  a quick overview of the important elements. Navigation through different zoom levels reveals the structure of the graph. 
In addition, users can interact with the system. For example, when the user clicks on a node $u$, the visualization highlights and renders all neighbors of $u$ (even those that do not belong to the current $G_i$) and the edges that connect $u$ to its neighbors. By using this interaction the user can explore a path  by selecting a set of successive nodes on the path, and can answer adjacency questions by selecting the corresponding pair of nodes.   
 
%%%%%%%%%%%%%%%%%%%%%%%%%%%%%%%%%%%%%%%%%%%
%USAGE SCEANRIO
%%%%%%%%%%%%%%%%%%%%%%%%%%%%%%%%%%%%%%%%%%%

We draw the nodes as points, and edges as polygonal chains. Each maximal straight line segment in the drawing is called a \emph{rail}. The edges may share rails. %, but no edge passes through a node.
 Every point where a pair of rails meet is either a node  or a point which we call a \emph{junction}. Figure~\ref{fig:introduction}(a) depicts a traditional node-link diagram of a graph $G$. Figures~\ref{fig:introduction}(b--c) illustrate a GraphMaps visualization of $G$ on three zoom levels. The gray region at each level corresponds to a  viewport in that level. The higher ranked nodes of $G$ have the darker color. The tiny gray dots represent the locations of the nodes that are not visible in the current layer. Figures~\ref{fig:zoom} illustrate the node selection technique and the zoom feature. The rails rendered by thick lines correspond to the shortest paths from the selected nodes $a$ and $h$ to their neighbors. %NEED TO MENTION THE COLOR TRANSPERENCY

\begin{figure}[h]
  \centering
  \includegraphics[width=\columnwidth]{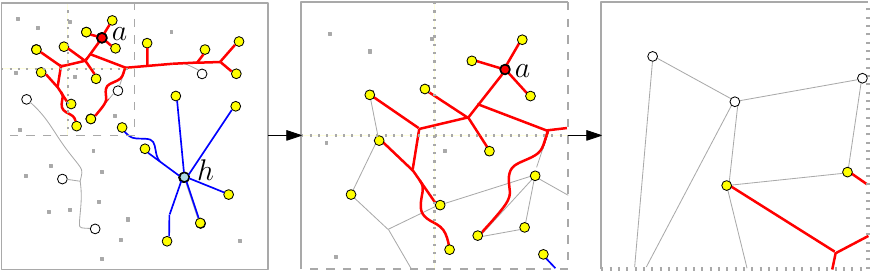}  
  \caption{Node selections and zoom in.%Illustration for node selections and zoom in. 
}
   \label{fig:zoom}
\end{figure}

In our scheme, where we change the rendering depending on the zoom level, the quality of the visualization depends both on the quality of the drawing on each level, and the differences between the drawings of successive zoom levels. We think that a good drawing of a graph on a single zoom level satisfies the following properties:
\begin{enumerate}
%\item The edges are drawn with small number of bends, i.e., the number of rails per edge, is small.
\item[-] The angular resolution is large
\item[-] The \emph{degree} at a node or at a  junction is small. 
\item[-] The amount of \emph{ink}, that is the sum of the lengths of all distinct rails used in the drawing, is small. 
\item[-] The \emph{edge stretch factor} or \emph{dilation}, that is the ratio of the length of an edge route to  the  Euclidean distance between its end nodes, is small.  
\end{enumerate}

\noindent
These properties help to follow the edge routes, reduce the visual load, and thus improve the readability of a drawing. Since some of the principles contradict each other, optimizing all of them simultaneously is a difficult task. % In Appendix A, we explain these challenges through a toy example of GraphMaps. 

Our algorithm, in addition to creating a  good drawing of each $G_i$, attempts to   construct these drawings in a way that a  switch from $G_i$ to $G_{i+1}$ does not cause a large change on the screen. We try to keep the amount of new appearing details relatively small and also try to keep the edge geometry stable. %We created a  \href{https://www.youtube.com/watch?v=qCUP20dQqBo&feature=youtu.be}{video}~\footnote{\href{https://www.youtube.com/watch?v=qCUP20dQqBo&feature=youtu.be}{https://www.youtube.com/watch?v=qCUP20dQqBo-\&feature=youtu.be}} to  demonstrate GraphMaps based on our approach.

\begin{figure*}[pt]   
\centering
\smallskip 
{\includegraphics[width=.3\textwidth, height = 3cm]{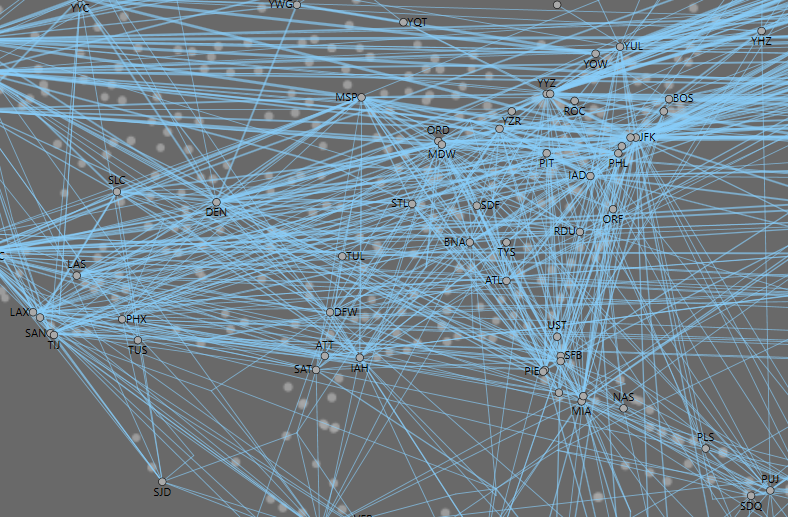}}
\hfill
{\includegraphics[width=.3\textwidth, height = 3cm]{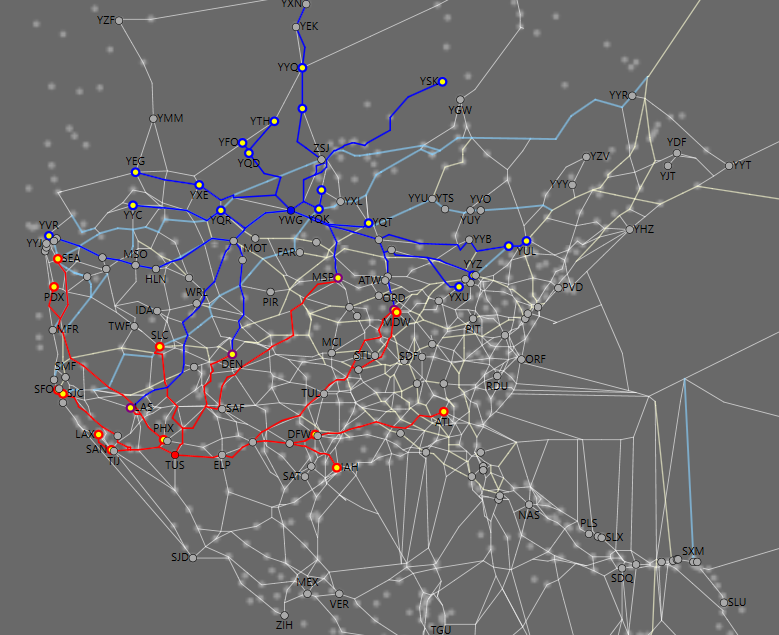}}
\hfill
{\includegraphics[width=.3\textwidth, height = 3cm]{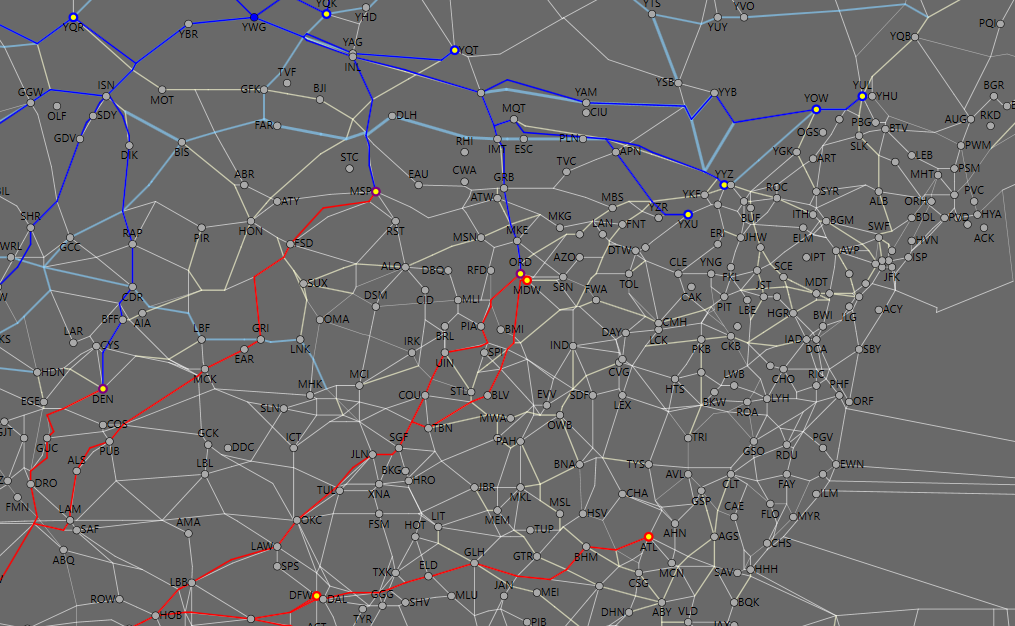}}
\caption{A partial display of a flight network with approximately 3K nodes and 19K edges. (left) Traditional node-link diagram over the airports of North America. (middle) The  top-level of a GraphMaps visualization based on our approach, where the airports TUS and YWG are selected. (right) A view after zoom in.}
\label{abs}
\end{figure*} 
%\noindent
\subsection{Related Work}   A large number of   graph visualization tools,
 e.g.,  Centrifuge, Cytoscape, Gephi, Graphviz, Biolayout3d, %~\cite{Centrifuge,Cytoscape,Gephi,Graphviz,biolayout3d}
 have been developed over the past few decades due to a growing interest in exploring network data. %Many of these tools provide a rich interaction set, which is useful to visualize large size graphs. %moderate  size graphs. %On the contrary, only a few tools support large graphs, e.g., graphs having several thousands of nodes and edges. One obvious reason for that is that the naive approach to rendering does not work any more. After rendering all the entities at once, it becomes very difficult to interact with the visualization.
 A good visualization requires the careful placement of nodes, % to be placed in some meaningful locations,
  e.g., sometimes nodes with  similar properties are  placed close to each other, whereas the nodes that are dissimilar are  placed far away. Force directed approaches, multi-dimensional scaling and stochastic neighbor embedding are some common techniques to  generate the node positions~\cite{Hu,KlimentaB12,MaatenH08}.   Techniques that try to make the  visualization readable by drawing the  edges carefully  consider various types of edge bundling~\cite{DBLP:journals/tvcg/ErsoyHPCT11,GansnerHNS11,LambertBA10,PupyrevNBH11} and edge routing techniques~\cite{DobkinGKN97,HoltenW09,DwyerN09}.   
%DwyerMW06
  Informally, the edge bundling technique groups the edges that are travelling towards a common direction, and routes these edges through some narrow tunnel. Figure~\ref{fig:introduction}(e) illustrates an example of edge bundling. %shows how edge bundling could be used to enhance the clarity of the visualization of Figure~\ref{fig:introduction}(a).  
  Other forms of clutter reduction approaches include node aggregation~\cite{DBLP:conf/chi/Wattenberg06,DBLP:conf/chi/DunneS13,DBLP:journals/tvcg/ZinsmaierBDS12}, topology compression~\cite{DBLP:conf/bigdataconf/ShiLSCL13,DBLP:journals/jgaa/BrunelGKRW14}, and sampling algorithms~\cite{DBLP:conf/globecom/GaoHL15}. 
  
This paper focuses on GraphMaps, proposed by Nachmanson et al.~\cite{Nachmanson15}, that reduces clutter by distributing nodes to different zoom levels and routing edges on shared rails.  Like the clutter reduction approaches, a primary goal of GraphMaps is to make the visualization more readable and interactive in the higher levels of abstraction. Nachmanson et al.~\cite{Nachmanson15}  use multidimensional scaling  to create the node positions. To distribute the nodes into zoom levels, they consider at each level $i$, an uniform   $2^i\times 2^i$ grid, where each grid cell is called a \emph{tile}. The tiles are filled with nodes, the most important nodes first. While filling the levels with nodes, they maintain a node and a rail quota that bound the number of nodes and rails intersecting a tile.  Whenever an insertion of a new node creates a tile intersecting more than one fourth of the node quota nodes or more than one quarter of rail quota rails, a new zoom level is created to insert the rest of the entities. The visualization of GraphMaps works in such a way that each viewport is covered by four tiles of the current level. This ensures that not more than the node quota nodes and the rail quota rails are rendered per viewport.

 %The  GraphMaphs visualizations  of~\cite{Nachmanson15} are expected to optimize the objectives O$_1$--O$_8$, but they do not give any theoretical guarantee on any of these measures.

GraphMaps visualization also relates to the hierarchical visualization of  clustered networks~\cite{DBLP:journals/tochi/SchafferZGBDDR96,DBLP:conf/apvis/BalzerD07}. We refer the reader to the survey~\cite{eurovisstar.20151110} for more details on visualizing graphs based on graph partitioning. There exist some systems that render large graphs on multiple layers by using the notion of temporal graphs, e.g., evolving software systems~\cite{CollbergKNPW03,LambertBA10}. A generalization of stochastic neighbor embedding renders nodes on multiple maps~\cite{MaatenH12}. Gansner et al.~\cite{GansnerHK10} proposed a visualization that emphasizes node clusters as geographic regions.   %However, all these approaches are different from GraphMaps in the way they visualize a network (see~\cite{Nachmanson15}  for further details).
 %with respect to the goals a GraphMaps system tries to achieve.  

%\subsubsection*
\subsection{Contribution} The existing implementation of GraphMaps~\cite{Nachmanson15} focuses mainly on the quality of the layout at individual zoom level.
 The construction follows a top-down approach, where the successive levels were obtained by inserting nodes incrementally in a greedy manner. %The edges were routed using Shewchuk's mesh generator~\cite{Shewchuk14},  where the maximum degree of the mesh could be as large as 12, edge stretch factor is roughly 2.4~\cite{BoseK06}, and the angular resolution is  $27^\circ$. {\color{red} is it a bad thing? It is enough to mention the angular resolution only, since the degree will follow. I think we need to spell out the contributions better here, maybe create a list. Can we say something about the number of layers that we generate now? The fact that we create layers bottom up is not a contribution:-)}

We propose an algorithm to construct a GraphMaps visualization  starting from a complete drawing of the graph $G(=G_k)$ at the bottom level. Specifically, given an arbitrary mesh and the edge routes of $G_k$ on this mesh, our method builds the edge routes for $G_{k-1},\ldots, G_{2}, G_1$, in this order.
 %
%We propose an algorithm to construct GraphMaps visualizations  starting from a complete drawing of the graph $G$ at the bottom level, i.e.. The method starts from any given mesh, and edge routes of graph $G (=G_k)$ routed on this mesh, and builds the edge routes for $G_i$ for $1 \le i < k$.
 %
We introduce a particular type of mesh, called competition mesh, which is of independent interest due to its low edge stretch factor $(2+\sqrt{2})$, and high angular resolution $45^\circ$. We then construct GraphMaps visualizations by applying our algorithm to this mesh. 
We develop a node assignment algorithm that minimizes the change in the drawing when switching from $G_i$ to $G_{i+1}$ during zoom in, where $1\le i <k$. %the zoom in operation, where $1\le i <k$.
Moreover, we propose new node selection techniques to cope with some of the challenges of the GraphMaps approach. 

%Although the primary goal of this paper is to advance the theory and develop an algorithmic framework for GraphMaps visualization, we

We also carried out experiments on some real-life datasets (see Figure~\ref{abs} and  Section~\ref{EXPERIMENTS}). Our experiments reveal the usefulness of GraphMaps, even in its basic implementations,  for understanding the network information through interactive exploration.
 
\section{\uppercase{Technical Background}}
\label{section:TB}

%\subsection{Construction of Meshes}
We now introduce the mesh that we use for edge routing and analyze its properties.  Let $P$ be a set of $n$ distinct points that correspond to the node positions, and let $R(P)$ be the smallest axis aligned rectangle that encloses all the points of $P$. A \emph{competition mesh} of $P$ is a geometric graph constructed by shooting from each point, four axis-aligned rays at the same constant speed (towards the top, bottom, left and right), where each ray stops as soon as it hits any other ray or $R(P)$.  We break the ties arbitrarily, i.e., if two non-parallel rays hit each other simultaneously, then arbitrarily one of these rays stops and the other ray continues.  If two rays are collinear  and hit each other from the opposite sides, then both rays stop. We denote this graph by $M(P)$, e.g., see Figure~\ref{fig:competitionmesh}. The vertices of $M(P)$ are the points of $P$ (\emph{nodes}), and the points where a pair of the rays meet (\emph{junctions}).   Two vertices in $M(P)$ are adjacent if and only if the straight line segment connecting them belongs to $M(P)$. A competition mesh can also be viewed as a variation of a motorcycle graph~\cite{EppsteinGKT08}, or a geometric spanner with Steiner points~\cite{BoseS13}. %Figure~\ref{fig:competitionmesh} illustrates a competition mesh of  8 points. Throughout this paper we use the term `vertices' to denote all the nodes and junctions of the mesh.
 In the rest of the paper the term `vertices' denotes all the nodes and junctions of $M(P)$.

\begin{figure*}[pt]
  \centering
  \includegraphics[width=.75\textwidth]{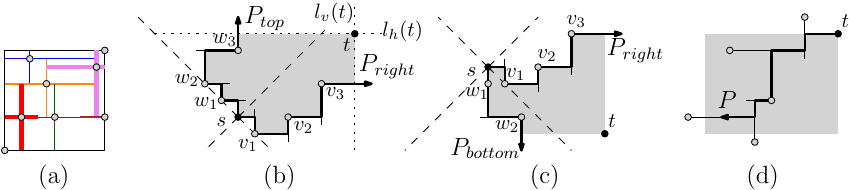}  
  \caption{(a) A point set and its corresponding competition mesh. (b--c) Bounding the bottom-left quadrant of $t$. (d) A monotone path inside the bottom-left quadrant of $t$.
  }
  \label{fig:competitionmesh}
 \end{figure*}

For any point $u$, let $u_x$ and $u_y$ be the $x$ and $y$-coordinates of $u$, respectively, and  let $l_v(u)$ and $l_h(u)$ be the vertical and horizontal straight lines through $u$, respectively.   For any two points $p,q$ in $\mathbb{R}^2$, let $dist_E(p,q)$ %(respectively, $dist_M(p,q)$)
  be the Euclidean distance %(respectively, Manhattan distance)
 between $p$ and $q$. For each point $u$ of the plane we define four \emph{quadrants}  formed by the horizontal and vertical lines passing through $u$. A path $v_1,\dots, v_k$ is \emph{monotone} in direction of vector $\textbf{a}$ if for each $1 \le i < k$ the dot product $\textbf{a}\cdot (v_{i+1} - v_i)$ is not negative. Lemmas~\ref{monotone}--\ref{dilation} prove some properties of $M(P)$. %The proof of Lemma~\ref{monotone} is included in Appendix A.

\begin{lemma} 
\label{monotone}
Let $u$ be a node in $M(P)$. Then in each quadrant of $u$ there is a path in $M(P)$ that starts at $u$, ends at some point on $R(P)$, and is monotone in both horizontal and vertical directions.  
\end{lemma}
\begin{proof}
Without loss of generality it suffices to prove the lemma for  the first quadrant of $u$, which consists of the set of points $v$ such that $v_x\ge u_x$ and $v_y\ge u_y$.   Suppose for a contradiction that there is no such path in this quadrant. Consider a maximal $xy$-monotone path $\Pi$ that starts at $u$ and ends at some node or junction $w$ of $M(P)$. If $w$ is a node, then we extend $\Pi$ using the right or top ray of $w$, which is a contradiction. Therefore, $w$ must be a junction in $M(P)$. Without loss of generality assume that the straight line segment $\ell$ incident to $w$ is horizontal. Since $\Pi$ is a maximal $xy$-monotone path, the ray $r_\ell$ corresponding to $\ell$  must be stopped by some vertical ray $r'$ generated by some vertex $w'$. Observe that $w_y> w'_y$, otherwise we can extend $\Pi$ towards $w'$. Since $r_\ell$ is stopped, the ray $r'$ must continue unless there are some other downward ray $r''$ that hits $r'$ at $w$. In both cases we can extend $\Pi$, either by following $r'$ (if it continues), or following the source of $r''$ (if $r'$ is stopped by $r''$), which contradicts to the assumption that $\Pi$ is maximal. 
\end{proof}

%\begin{proof}
%Without loss of generality it suffices to prove the lemma for  the first quadrant of $u$, which consists of the set of points $v$ such that $v_x\ge u_x$ and $v_y\ge u_y$.   Suppose for a contradiction that there is no such path in this quadrant. Consider a maximal $xy$-monotone path $\Pi$ that starts at $u$ and ends at some node or junction $w$ of $M(P)$. If $w$ is a node, then we extend $\Pi$ using the right or top ray of $w$, which is a contradiction. Therefore, $w$ must be a junction in $M(P)$. Without loss of generality assume that the straight line segment $\ell$ incident to $w$ is horizontal. Since $\Pi$ is a maximal $xy$-monotone path, the ray $r_\ell$ corresponding to $\ell$  must be stopped by some vertical ray $r'$ generated by some vertex $w'$. Observe that $w_y> w'_y$, otherwise we can extend $\Pi$ towards $w'$. Since $r_\ell$ is stopped, the ray $r'$ must continue unless there are some other downward ray $r''$ that hits $r'$ at $w$. In both cases we can extend $\Pi$, either by following $r'$ (if it continues), or following the source of $r''$ (if $r'$ is stopped by $r''$), which contradicts to the assumption that $\Pi$ is maximal. 
%\end{proof}
  
\begin{lemma} 
\label{dilation}
For any set $P$ with $n$ points  $M(P)$ has $O(n)$  vertices and  edges. The graph distance  between any two nodes of $M(P)$ is at most $ (2+\sqrt{2})$ times the Euclidean distance. %Furthermore, $M(P)$  can be constructed in $O(n \log n)$ time. 
\end{lemma}
\begin{proof}
By construction of $M(P)$,  whenever a junction is created, one ray stops. Since $|P|=n$ there are at most $4n$ rails, and therefore we cannot have more than $4n$ junctions, that proves that the number of vertices in $M(P)$ is $O(n)$. Since $M$ is a planar graph, the number of edges is also $O(n)$. 

We now show that the ratio of the graph distance and the Euclidean distance between any two nodes of $M(P)$ is at most $(2+\sqrt{2})$. Let $C_{left}, C_{right},C_{top}$, and $C_{bottom}$ be the four   cones with the apex at $(0,0)$ determined by the lines $y=\pm x$.  Let $s$ and $t$ be two nodes in $M(P)$.  Without loss of generality assume that $s$ located at $(0,0)$ and $t$ lies on $C_{right}$. Consider now an $x$-monotone orthogonal path $P_{right}=(v_0,v_1,v_2,\ldots,v_q)$ in the mesh such that $v_0$ coincides with $s$, for each $0<i\le q$, $v_i$ is a node in $M(P)$ that stops the rightward ray of $v_{i-1}$, and $v_q$ lies on or to the right of $l_v(t)$, e.g., see  Figure~\ref{fig:competitionmesh}(b).   Suppose that $t$ is either above or below $P_{right}$. If $t$ is above $P_{right}$, then  consider a $y$-monotone path $P_{top}=(w_0,$ $w_1,w_2,\ldots,w_r)$ such that  $w_0$ coincides with $s$, for each $0<i\le r$, $w_i$ is a node in $M(P)$ that stops the upward ray of $v_{i-1}$, and $v_r$ lies on or above  $l_h(t)$.  If $t$ is below $P_{right}$, then define a path $P_{bottom}$ symmetrically, e.g., see Figure~\ref{fig:competitionmesh}(c). 

Without loss of generality assume that $t$ is above $P_{right}$. Observe that the paths $P_{top}$ and $P_{right}$ remain inside the cones $C_{top}$ and $C_{right}$, respectively, and  bound the bottom-left quadrant of $t$, as shown in the shaded region in Figure~\ref{fig:competitionmesh}(b). By Lemma~\ref{monotone}, $t$ has a  $(-x)(-y)$-monotone path $P$ that starts at $t$ and reaches the boundary of $R(P)$, e.g., see Figure~\ref{fig:competitionmesh}(d). This path $P$ must intersect either $P_{top}$ or $P_{right}$. Hence we can find a path $P'$ from $s$ to $t$, where $P'$ starts at $s$, travels along either $P_{top}$ or  $P_{right}$ depending on which one $P$ intersects, and then follows $P$ from the intersection point. We now show that length of $P'$  is at most  $(2+\sqrt{2})\cdot dist_E(s,t)$. 
Since any ray is not shorter than a ray it stops, the sum of the lengths of the vertical segments of $P_{right}$  is at most the sum of the lengths of the horizontal segments. Therefore, the part of $P_{right}$ inside the bottom-left quadrant of $t$ is at most $2t_x$. Similarly, the part of $P_{top}$ inside the bottom-left quadrant of $t$ is at most $2t_y$. Path $P$ is not longer than $t_x+t_y$ (see Figure~\ref{fig:competitionmesh}(d)) . Therefore, the length of $P'$ is at most  $t_x+t_y + 2 \cdot\max\{    t_x, t_y)\}$
 $\le \sqrt{2} \cdot dist_E(s,t) +2  \cdot\max\{ t_x, t_y\}$
 $\le (2+\sqrt{2}) \cdot dist_E(s,t)$.

%\begin{align*}
%& t_x+t_y + 2 \cdot\max\{    t_x, t_y)\} \\
%&\le \sqrt{2} \cdot dist_E(s,t) +2  \cdot\max\{ t_x, t_y\} \\
%& \le (2+\sqrt{2}) \cdot dist_E(s,t).
%\end{align*}
In the case when $t$ belongs to $P_{right}$,  the length of path $P$ is zero and the proof easily follows.
 %The proof that $M(P)$ can be constructed in $O(n \log n)$ time is included in Appendix A. 
\end{proof}

The following lemma states that a competition mesh can be constructed in $O(n\log n)$ time.
%, whose proof is included in the Appendix B.
\begin{lemma} 
\label{lem:time}
For any set $P$ with $n$ points, the competition mesh $M(P)$ can be constructed in $O(n \log n)$ time. 
\end{lemma}
\begin{proof}
Define for each point $w\in P$, a set of eight non-overlapping  cones as follows: The central   angle of each cone is $45^\circ$  and the cones are ordered counter clockwise around $w$. The first cone lies in the first quadrant of $w$  between the lines $y=x+w_x$ and $y=0$, as shown in Figure~\ref{fig:construction}(a).   Guibas and Stolfi~\cite{GuibasS83} showed that in $O(n\log n)$ time, one  can find for every point $w\in P$ the nearest neighbor of $w$ (according to the Manhattan Metric) in each of the eight cones of $w$. Assume that $\delta_y = \{ \min_{\{a,b\}\in P \text{, where } a_y \not= b_y} |a_y - b_y|  \}$, which can be computed in $O(n \log n)$ time by sorting the points according to $y$-coordinates.

\begin{figure*}[pt]
  \centering
  \includegraphics[width=.7\textwidth]{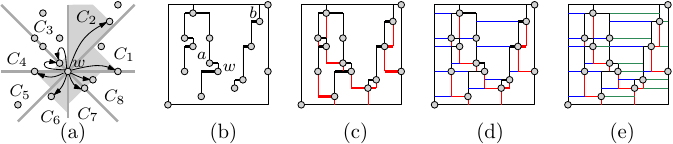}  
  \caption{(a) A point set $P$ and the nearest neighbor  of $w$ in each of the eight cones around $w$. (b--e) Construction of the mesh of $P$, while processing (b) top rays, (c) bottom rays, (d) left rays and (e) right rays.}
  \label{fig:construction}
\end{figure*} 
 
We construct $M$ in four phases. The first phase iterates through the top rays of the each point $w$ and   finds the point $w'$ closest to $l(h)_w$ (in Manhattan metric) such that  $|w'_x-w_x|\le |w'_y-w_y|$. Note that if  such a $w'$ exists,  then one of the horizontal rays  $r'$ of $w'$ would reach the point $(w_x,w'_y)$ before the top ray $r$ of $w$ (while all rays are grown in an uniform speed). According to the definition of the competition mesh, we can continue the ray $r'$ and stop growing the ray $r$.  If no such $w'$ exists, then $r$ must hit $R(P)$.

To find the point $w'$, it suffices to compare the Manhattan distances of the nearest neighbors in the second and third cones of $w$ (breaking ties arbitrarily). Since the nearest neighbors at each cone can be accessed in $O(1)$ time, we can process all the top rays in linear time. Figure~\ref{fig:construction}(b) shows the junctions and nodes created after the first phase, where all the rays that stopped growing are shown in thin lines.  The nearest neighbors in the second and third cones of $w$ are $a$ and $b$, respectively. Since  $a$ is closer to $l_h(w)$ than $b$, the top ray of $w$ does not grow beyond $(w_x,a_x)$. The second   phase  processes the bottom  rays  in a similar fashion, e.g., see Figure~\ref{fig:construction}(c). Both the first and second phase use a ray shooting data structure to check whether the current ray already hits an existing ray. We describe this data structure in the next paragraph while discussing phase three.  Let the planar subdivision at the end of phase two be $S$.

In the third phase, we   grow  each left ray until it hits any other  vertical segment in $S$, as follows:  For each vertical edge $ab$ in $S$, construct a segment $a'b'$ by shrinking $ab$ by $\delta_y/3$ from both ends. For each node and junction $w$, construct a segment $w_1w_2$ such that $w_1 = (w_x,w_y-\delta_y/4)$ and $w_2 =  (w_x,w_y+\delta_y/4)$. Let $S_v$ be the constructed segments.  Note that the segments in $S_v$ are non-intersecting. Giyora and Kaplan~\cite{GiyoraK09} gave a ray shooting data structure $D$ that can process $O(n)$ non-intersecting vertical rays in  $O(n\log n)$ time, and given   a query point $q$,  $D$ can find in $O(\log n)$ time the  segment (if any)  in $S_v$ immediately to the left of $q$.  Furthermore, $D$ supports insertion and deletion in $O(\log n)$ time. For each point $q\in P$ in the increasing order of $x$-coordinates, we shoot a leftward  ray $r$ from $q$, and find the first segment $ab$ hit by the ray. Assume that  $a_y<b_y$ and $r$ hits $ab$ at point $x$. We update the subdivision $S$ accordingly,  then delete segment $ab$ from $D$,  and insert  segment $xb$ in $D$. Note that these updates keep  all the segments in $D$ non-intersecting. Since there are $O(n)$ left rays, processing all these rays takes $O(n \log n)$ time. Figure~\ref{fig:construction}(d) illustrates the third phase.

The fourth   phase processes the right  rays  in a similar fashion, e.g., see  Figure~\ref{fig:construction}(e).  Since the preprocessing time of the data structures we use is $O(n\log n)$, and since each phase runs in $O(n \log n)$ time,  the construction of the computation mesh takes $O(n \log n)$ time in total. 
\end{proof}
%In the subsequent sections, we describe how the competition mesh is used in GraphMaps.

\section{\uppercase{GraphMaps System}} % Based on Competition Mesh}}
\label{sec:cm}
Our technique for calculating the graphs $G_1,\dots, G_k$ (equivalently, node level assignment) is described in Section~\ref{sec:smooth}. For now, let us assume  that the sequence of graphs is ready. We now show how to route edges on graphs $G_i$.

%\subsection{Edge Routing} % and Local Modifications}
%\label{sec:interpolation}
The computation of edges starts from the bottom. Namely we build a competition mesh $M$ for graph $G (=G_k)$.  We route each edge $(u,v)\in G$ as a shortest path $P_{uv}$ in $M$. 
%Path $P_{uv}$ should not pass through another node of $G$, otherwise the visualization would show some edges that do not exist in $G$. To achieve this property, we surround each node of $G$ by a polygon that creates a detour for $P_{u,v}$, and while routing, forbid  $P_{u,v}$ to enter a node distinct from $u$ and $v$. Hence, the routing in fact runs on the modified $M$. Figures~\ref{fig:detour}(a--c) illustrate an example where $P_{uv}$ avoids the nodes $a$ and $b$. After routing all edges of $G$, we remove from $M$ every edge that is not used by a route $P_{u,v}$ for any edge $(u,v)$ of $G$. %The worst case time complexity of computing all pair shortest paths in a planar graph is $O(n^2)$~\cite{HenzingerKRS97}, and computing the shortest paths is the main bottleneck of the time complexity. 
% Computing the shortest paths is the main  time-complexity bottleneck while constructing GraphMaps.  We use A* search to find the shortest paths, which is faster in practice. %which is fast enough for ... nodes and edges
 Let us denote by $M'$ the mesh we obtain after applying these modifications to $M$. Next we modify $M'$ to make the routes more visually appealing. We perform local modifications and try to minimize the total ink of the routes, which is the sum of lengths of edges of $M'$ used in the routes~\cite{gansner2006improved}, and remove thin faces. During the modifications we keep the angular resolution greater or equal than some $\alpha >0$, and the minimum distance between non-incident vertices and edges of $M'$ greater or equal than some $\beta>0$. The local modifications are described below.

   \begin{figure}[pb]
  \centering
  \includegraphics[width=.4\textwidth]{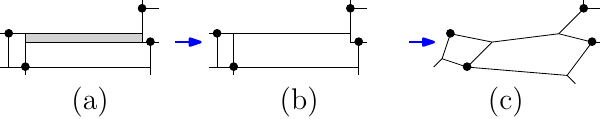}  
  \caption{(a--b) Removal of thin faces. (b--c) Moving junctions towards median. %geometric median.  
  }
  \label{fig:detour}
 \end{figure} 
 
   \begin{figure*}[pt]
  \centering
  \includegraphics[width=.75\textwidth]{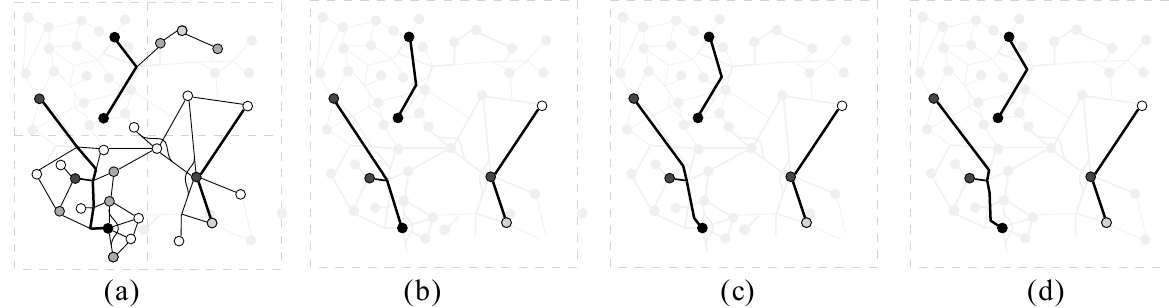}  
  \includegraphics[width=.75\textwidth]{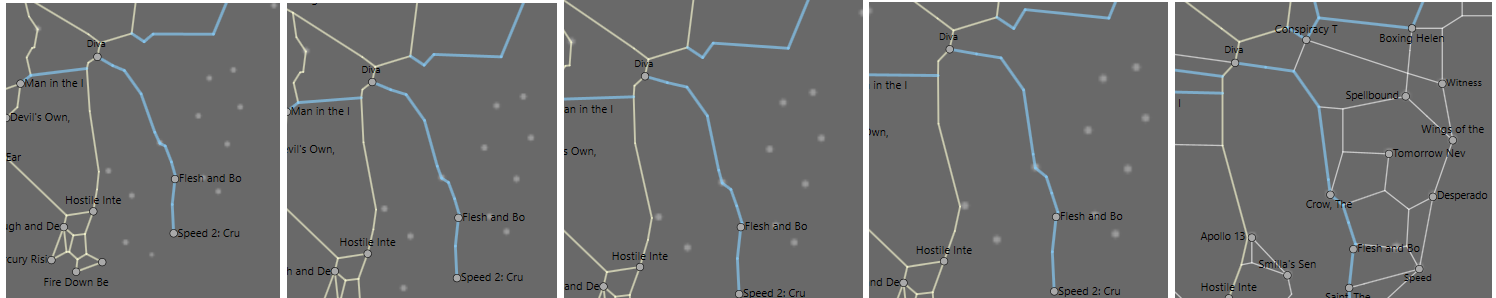}  
  \caption{(a) Zoom level 2. (b--d) Transition from level 1 to 2.
  (bottom) Transition in our GraphMaps system.
  }
  \label{fig:pathsim}
 \end{figure*} 

\smallskip
%\begin{description}
%\item[Face refinement:] 
\textit{Face refinement:} For each face $f$ of $M'$ that does not contain a node of $G$ in its boundary, we compute the width of $f$, which is the smallest Euclidean distance between any two non-adjacent rails of $f$. If the width of $f$ is smaller than some given threshold, then remove the longest edge of $f$ from $M'$ (breaking ties arbitrarily). Figures~\ref{fig:detour}(a--b) depict such a removal, where the thin face is shown in gray. The edge routes using the removed edge are rerouted through the remaining boundary of $f'$.

%\item [Median:] 
\textit{Median:} Move each junction $\kappa$ of $M'$ toward the geometric median of its neighbors, i.e., the point that minimizes the sum of distance to the neighbors, as long as the restriction mentioned above holds.  Iterate the move for a certain number of times, or until the change becomes smaller than some given threshold. Figures~\ref{fig:detour}(c--d) illustrate the  outcome of this step.

%\item [Shortcut:] 
\textit{Shortcut:} Remove every degree two junction and replace the two edges adjacent to it by the edge shortcutting the removed junction, as long as the restriction mentioned above holds. 
%\end{description}
\smallskip

In all the above modifications, the routes are updated accordingly. Modifications ``Median'' and ``Shortcut'' diminish the ink.   The final $M'$ gives the geometry of the bottom-level drawing of $G$ in our version of GraphMaps.

%Finally, we insert the nodes greedily according to their priority, where the high priority nodes appear on the top levels. We start inserting nodes in the successive zoom levels only when any more insertion in the current zoom level violates the node quota. Given an upper bound on the zoom level, we can find the minimal node quota  $Q_n$,  such that no tile contains more than $Q_n$ nodes, in $O(\log n)$ iterations.

%Path Simplification and 
%\subsection{Transition between Levels}
 Given a mesh $M_i$ representing the geometry for the drawing of $G_i$, where $1<i\le k$, we construct $M_{i-1}$ from $M_i$ by removing from the latter the nodes $V(G_i) \setminus V(G_{i-1})$, and by removing the edges that are not used by any route $P_{u,v}$, where $(u,v)$ is an edge of $G_{i-1}$.
Some routes $P_{u,v}$ can be straightened in $M_{i-1}$. %This modification is often impossible in $M_i$ because it would create a node-edge overlap.
 We use the simplification algorithm~\cite{pathsim} to morph the paths of $M_{i}$ to paths of $M_{i-1}$.  Figures~\ref{fig:pathsim}(a--b) illustrate the simplification.  

This change in the edge routes geometry diminishes the consistency between the drawings of successive levels. To smoothen the differences while transiting from zoom level $i$ to $i+1$ we linearly interpolate between the paths of $M_i$ and $M_{i+1}$, as demonstrated in   Figures~\ref{fig:pathsim}(b--d).

The idea of path simplification and transition via linear interpolation enables us to construct GraphMaps in a bottom-up approach.  In fact, the above  strategy can be applied to transform any mesh generated from a set of 2D points to a GraphMaps visualization. 

\section{\uppercase{Computing Node Levels}}
\label{sec:smooth}
Let us consider in more details how the view changes when we zoom  by examining Figures~\ref{fig:transition}(a)--(d). On the top-left tile of Figure~\ref{fig:transition}(a), the user's viewport covers the whole graph, so $G_1$ is exposed. 
 In  Figure~\ref{fig:transition}(d), the user's viewport contains only the top-left tile of Figure~\ref{fig:transition}(a), and the visualization switches to graph $G_2$. Seven new nodes, which were not fully visible in zoom level 1,   become fully visible for the current viewport, as represented in light gray. If all of a sudden, a large number of new nodes become fully visible, then it may disrupt user's mental map.  Here we propose an algorithm to keep this change small.

   \begin{figure*}[pt]
  \centering
  \includegraphics[width=.75\textwidth]{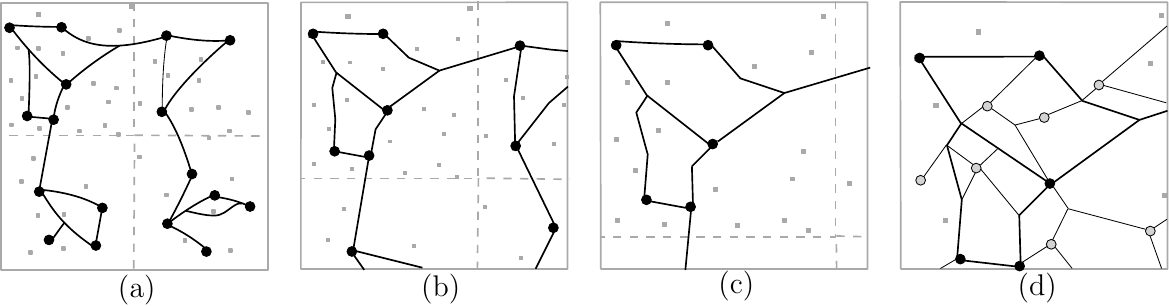}  
  \caption{(a) Zoom level 1. (b--c) Transition from level 1 to 2. (d) Zoom level 2.
  }
  \label{fig:transition}
 \end{figure*}   

We build the tiles as in~\cite{Nachmanson15}. In the first level we have only one tile coinciding with the graph bounding box. On the $i$th level, where $i>1$, the tiles  are obtained by splitting each tile in the $(i-1)$th level into a uniform $2\times 2$ grid cell. This arrangement of tiles can be considered as a rooted tree $T$, where the tiles correspond to the nodes of the tree. Specifically, the topmost tile is the root of $T$, and a node $u$ is a child of another node $v$ if the corresponding tiles $t_u$ and $t_v$ lie in two different but adjacent levels, and $t_u\subset t_v$. We refer to $T$ as a \emph{tile tree}.  

For every node $v$ in $T$, denote by $S(v)$ the number of fully visible nodes in the tile  $t_v$. For an edge $e=(v,w)$ in $T$, where $v$ is a parent of $w$, we denote by $\delta_e$ the number of new nodes that become visible while navigating from $t_v$ to $t_w$, i.e., $\delta_e = |S(w)\setminus S(v)|$. We can control the rate the nodes appear and disappear from the viewport  by  minimizing $\sum_{e\in E(T)}\delta_e^2$, where $E(T)$ is the set of edges in $T$.  For simplicity we show how to solve the problem in one dimension, where all the points are lying on a horizontal line. It is straightforward to extend the technique in $\mathbb{R}^2$.

\smallskip
\begin{enumerate}
\item[] \textit{Problem.} \textsc{Balanced Visualization}
\item[] \textit{Input.} A set $P$ of $n$ points  on a horizontal line, where every point $q\in P$ is assigned a rank $r(q)$. A  tile tree $T$ of height $\rho$; and a node quota $Q$, i.e., the number of points allowed to appear in each tile. 
\item[] \textit{Output.} Compute a mapping $g:P \rightarrow \{1,2,\ldots \rho\}$ (if exists) that 
\begin{enumerate} 
 \item[-] satisfies the node quota, 
%If the decision is affirmative, then find a mapping that
 \item[-] minimizes  the objective  $F = \sum_{e\in E(T)}\delta_e^2$, and 
 \item[-] for every pair of points $q,q'\in P$ with $r(q) \ge r(q')$, satisfies the inequality $g(q)\le g(q')$, which we refer to as the \emph{rank condition}. 
\end{enumerate}
\end{enumerate}
\smallskip

If the rank of all the points are distinct, then the solution to \textsc{Balanced Visualization} is unique, and can be computed in a greedy approach. But the problem becomes non-trivial when many points may have the same rank. In this scenario, we prove that the \textsc{Balanced Visualization} problem can be solved in $O(\tau^2\log^2 \tau )+O(n \log n)$ time, where $\tau$ is the number of nodes in $T$. This is quite fast since the maximum zoom level is a small number, i.e., at most 10, in practice. 
 We reduce the problem to the problem of computing a minimum cost maximum flow problem, where the edge costs can be convex~\cite{Orlin,OrlinV13}, e.g., quadratic function of the flow passing through  the edge. Figure~\ref{fig:flow}(a) depicts a set of points on a line, where the associated tiles are shown using rectangular regions. The numbers in each rectangle is the number of points in the corresponding tile.  Figure~\ref{fig:flow}(b) shows a corresponding network $G$, where the source is denoted by $s$, and the sinks are denoted by $T_1,T_2,\ldots,T_8$. % are shown in shaded squares.{\color{red} I don't see shading in the source}.
  \begin{figure*}[h]
  \centering
  \includegraphics[width=.8\textwidth]{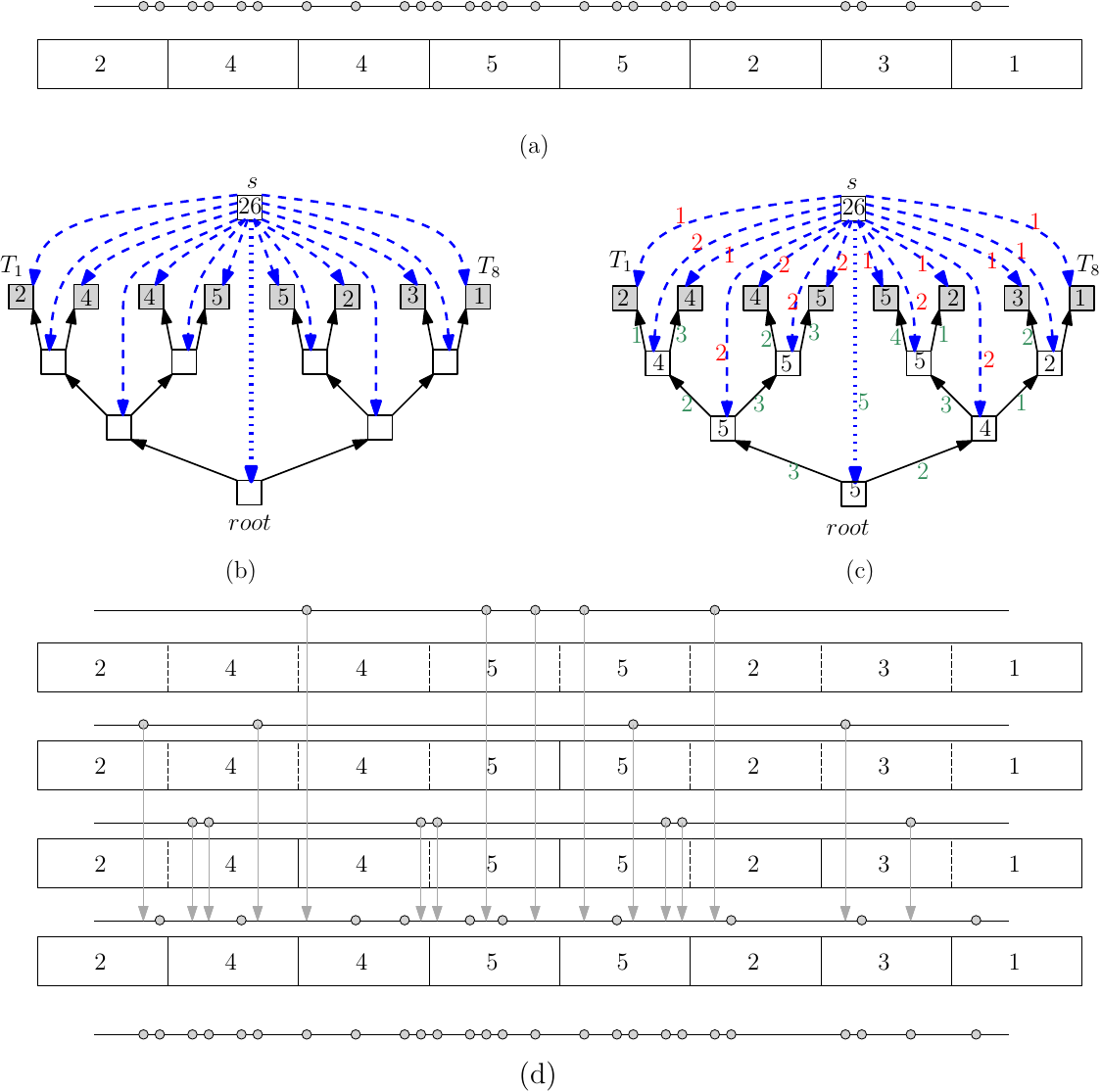}  
  \caption{(a) A set of points on a line and the associated tiles are shown using rectangular regions.  (b) A corresponding network $G$.   (c) A solution to the minimum cost maximum flow problem. (d) A solution to the \textsc{Balanced Visualization} problem that corresponds to the network flow.  %The blue, black, red and white points first appear on zoom level 4, 3, 2, 1, respectively.
  }
   \label{fig:flow}
\end{figure*} 
 The excess at the source and the deficit at the sinks are written in their associated squares.  We allow  each internal node $w$ (unfilled square)  of the tile tree to pass at most $Q$ units of flow through it. This can be modeled by replacing $w$ by  an edge $(u,v)$ of capacity $Q$, where all the edges incoming to $w$ are incident to $u$ and the outgoing edges are incident to $v$. This transformation is not shown in the figure. Set the capacity of all other edges  to $\infty$.

The production of the source is $n$ units, and the units of flow that each sink can consume is equal to the number of points lying in the corresponding tile.  % {\color{red} do we need this on all edges?}.
 The cost of sending flows along the tree edges (solid black) is zero. The cost of sending flows along the dotted edge connecting the source and the tree root is also zero. The cost of sending $x$ units of flow along the dashed  edges is $x^2$; sending $x$ units of flow through a dashed edge corresponds to $x$ new nodes that are becoming visible when we zoom in at the tile of the edge target. Figure~\ref{fig:flow}(c) illustrates a solution to the minimum cost maximum flow problem, where the flows are interpreted as follows:

\smallskip
\begin{enumerate}
\item[(A)]  The number in a square denotes  the number of points that would be visible in the associated  tile.
\item[(B)]  The edges $(s,w)$, where $w$ is not the root,  are labeled by numbers. Each such number corresponds  to the new nodes that will appear while zooming in from the  tile associated to $parent(w)$ to the tile associated to $w$. Thus the cost of this network flow is the sum of the squares of these numbers, i.e., $35$.
\item[(C)]  Each edge of $(u,v)$ of $T$ is labeled by the number of nodes that are fully visible both $u$ and in $v$.  
%The remaining labeled edges correspond to the number of points that are visible both in the parent tile and the child tile.
\item[(D)]  Any one unit source-to-sink flow  corresponds to a point of $P$, where the flow path $source, $ $u_1,u_2,\ldots,u_k(=sink)$ denotes that the point appeared in all the tiles associated to $u_1,\ldots,u_k$. 
\end{enumerate}
\smallskip

%\noindent
%The proof of the following lemma is included in Appendix B.
%We now have the following lemma, whose proof is omitted due to space constraints, but included in Appendix C. 
\begin{lemma}
\label{lem:flow}
A minimum cost maximum flow in $G$ minimizes the objective function $F$ of the \textsc{Balanced Visualization} problem.
\end{lemma} 
\begin{proof}
If the amount of flow consumed at sink is smaller than $n$, then we can find a cut in $G$ with total capacity less than $n$. Thus even if we saturate the corresponding tiles with points from $P$, we will not be able to visualize all the points without violating the node quota. Therefore, we can visualize all the points if and only if the flow is maximum and the total consumption is $n$. Therefore, the only concern is whether the solution with cost $\lambda$ obtained from  flow-network model  minimizes the sum of the squared node differences between every parent and child tiles. Suppose for a contradiction that there exists another solution of \textsc{Balanced Visualization} with cost $\lambda' < \lambda$.  In this scenario we can label the edges of the network according to the interpretation used in (A)--(D) to obtain a maximum flow in $G$ with cost $\lambda'$. Therefore,  the minimum cost computed via the flow-network model cannot be $\lambda$, a contradiction.
\end{proof}

Given a solution to the network flow, we can construct a corresponding solution to the \textsc{Balanced Visualization} problem as described below. 

\smallskip
\begin{enumerate}
\item[-] For each point $w$, set $g (w) = \infty$. 
\item[-] For each zoom level $z$ from $\rho$ to 1, process the tiles of zoom level $z$ as follows. Let $W$ be a tile in zoom level $z$. Find the amount of flow $x$  incoming to $W$ from $s$ in $G$. Note that this  amount $x$ corresponds to the difference in the number of points between $W$ and $parent(W)$. Therefore, we find a set $L$ of $x$  lowest priority points in $W$ with zoom level equal to $\infty$, then for each $w\in L$, we set $g(w) = z$. Figure~\ref{fig:flow}(d) illustrates a solution to the \textsc{Balanced Visualization} problem that corresponds to the network flow of Figure~\ref{fig:flow}(c).
\end{enumerate}
\smallskip

\noindent
If the resulting mapping does not satisfy the rank condition,  then the instance of \textsc{Balanced Visualization} does not have any affirmative solution.  The best known running time for solving a convex cost network flow problem on a network of size $O(\tau)$ is $O(\tau^2 \log^2 \tau)$~\cite{Orlin,OrlinV13}. Besides, it is straightforward to compute the corresponding node assignments in $O(n\log n)$ time augmenting the merge sort technique with basic data structures. 
 Hence we obtain the following theorem.
\begin{theorem}
Given a set of $n$ points in $\mathbb{R}$,  a tile tree of $\tau$ nodes, and a quota $Q$, one can find a balanced visualization (if exists) in $O(\tau^2 \log^2 \tau) {+} O(n\log n)$ time.
\end{theorem}

While implementing GraphMaps, we need to choose a node quota $Q$  depending on the given total number of zoom levels $\rho$. Using a binary search on the number of nodes, in $O(\log n)$ iterations, one can find a minimal node quota that allow visualizing all the points of $P$ satisfying the rank condition.

\section{\uppercase{Experiments}}
\label{EXPERIMENTS}
%\section{\uppercase{Comparison}} % with the previous version of GraphMaps}}
The GraphMaps system proposed previously~\cite{Nachmanson15} uses  %\href{http://www.cs.cmu.edu/~quake/triangle.html}{Triangle of J. Shewchook}
\footnote{\href{https://www.cs.cmu.edu/~quake/triangle.html}{\url{https://www.cs.cmu.edu/~quake/triangle.html}}} to obtain the graph for routing edges on a level. Our approach does not depend on Triangle, but uses Competition Mesh. This has several advantages. For example,  Competition Mesh usually produces less edges than Triangle. As a result the edge routing runs much faster.  With the same initial layout for the nodes, the GraphMaps system based on our approach sped up the initial processing significantly, up to 8 times on some graphs. The graph with 38395 nodes and 85763 edges was processed in the new system within 1 hour and 45 minutes, where the previous GraphMaps system took 6 hours~\cite{Nachmanson15}. Besides, Competition Mesh is more robust than Triangle. We did not experience failures, which were reported on Triangle~\cite{Nachmanson15}.
\begin{figure}[pt]
\centering 
{\includegraphics[width=.45\columnwidth, height = 3cm]{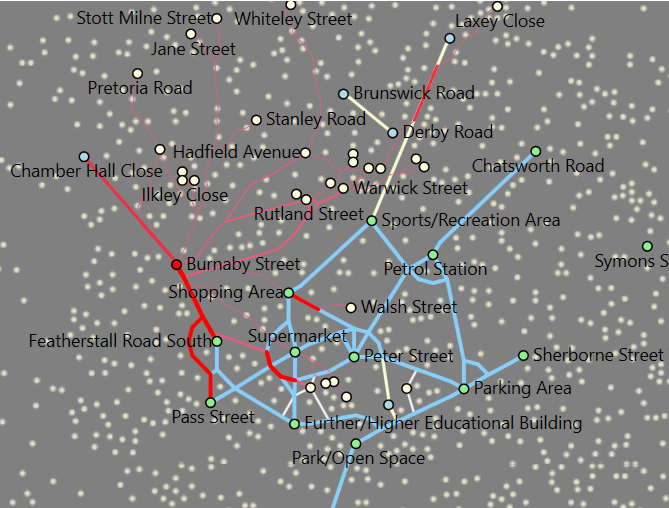}}
\hfill
{\includegraphics[width=.45\columnwidth, height = 3cm]{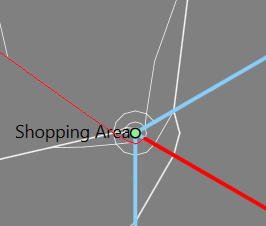}}
\caption{Node selection in  previous GraphMaps~\cite{Nachmanson15}. }
\label{det}
\end{figure} 
 
The previous GraphMaps~\cite{Nachmanson15} supports node selection, which is initiated by the user clicks. Selection of a node highlights the paths to its neighbors in red. This may create ambiguity. For example, Figure~\ref{det}(top-left) shows a graph of Burglary events (April 2015) in Manchester, UK, where two events are adjacent if they are located within 1km of each other. Selection of the node `Burnaby Street' highlights a rail very close to the node `Shopping area', which gives a false impression that these nodes are adjacent. After zooming in one can see that there is a detour that carries the highlighted path away from `Shopping area', e.g., see Figure~\ref{det}(top-right). This also give a wrong  impression of the node degree. Besides, since the edges may share rails, selection of multiple nodes may obscure the adjacency relationship, e.g., see Figure~\ref{pn}(left).

 \begin{figure}[pt]
\centering 
{\includegraphics[width=.45\columnwidth, height = 3cm]{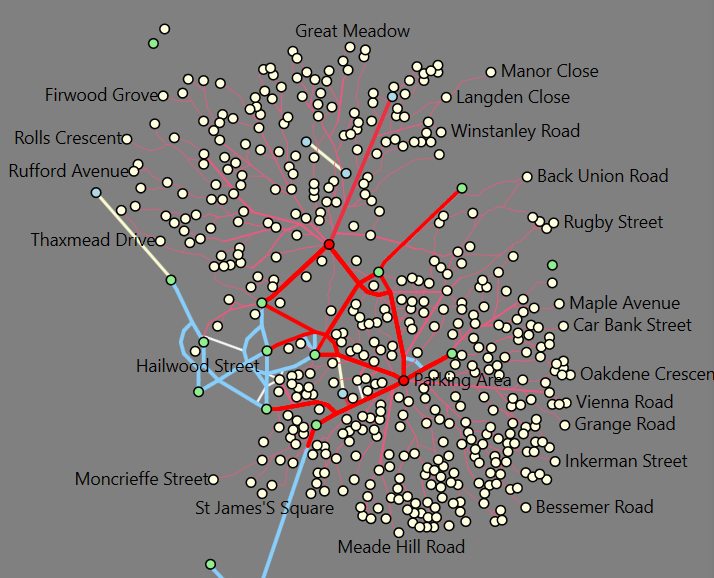}}
\hfill
{\includegraphics[width=.45\columnwidth, height = 3cm]{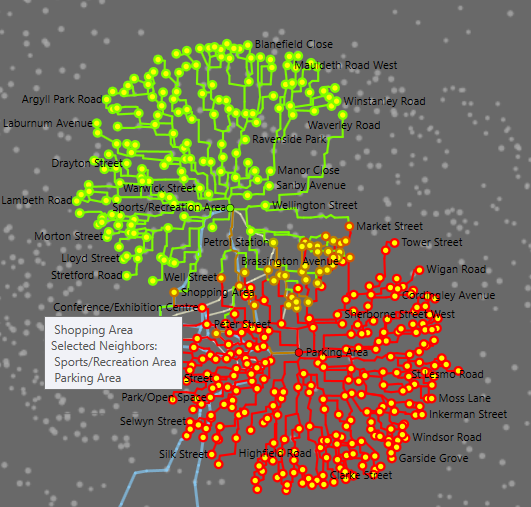}}
\caption{Selection of multiple nodes: (a) previous GraphMaps~\cite{Nachmanson15}, and (b) our approach.}
\label{pn}
\end{figure} 
 \begin{figure*}[h]
\centering 
{\includegraphics[width=\columnwidth, height = 6cm]{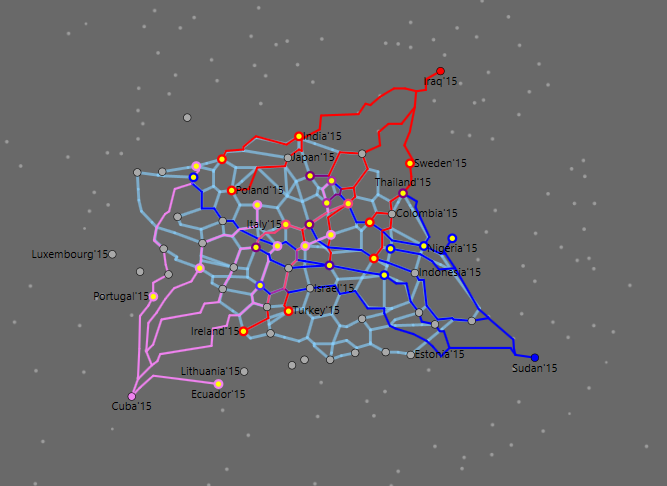}}
\hfill
{\includegraphics[width=\columnwidth, height = 6cm]{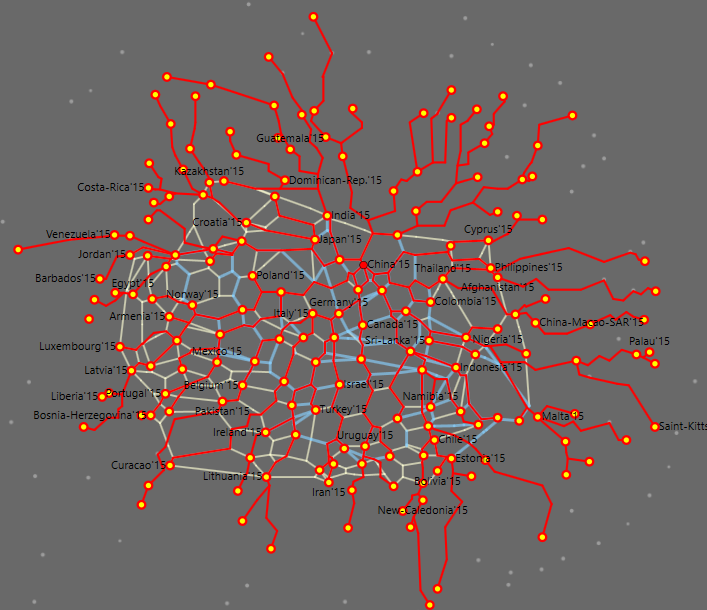}}
\caption{A visualization of gemstone trade-relation among countries ($\approx$7K edges) in 2015    (https://resourcetrade.earth). Selected nodes: (left) Iraq, Sudan, Cuba, and (right) China.}
\label{trade}
\end{figure*} 

%\section{\uppercase{Interactions and experiments}}
We introduce new visualizations that allow  the user to better understand the node neighborhood. Clicking on a node toggles its status from selected to not selected. When a node is selected, its neighbors are highlighted in yellow color, and the edges connecting the clicked node with the neighbors are highlighted with some unique color. If the mouse pointer hovers over a node highlighted by yellow color, then a tooltip appears with the list of the node neighbors, e.g.,  see Figure~\ref{pn}(right). When a selected node is unselected, then every edge adjacent to it is rendered in the default color, and the highlighting is removed from each neighbor, unless it is a neighbor of another selected node. These visualization measures help to resolve some ambiguities caused by edge bundling.  
 Note that our approach does not create any detour, and thus avoids the circular artifacts (rails) around the nodes. Exploring the node adjacencies and degree becomes comparatively easy, and less number of rails aids faster level transition.

Like the previous GraphMaps, our approach can also revel the structural properties of the graph. 
 Figure~\ref{trade} depicts a gemstone trade graph, where the countries with most trade relation, e.g., China, are in the central position and the countries with small number of trade relations fall into the periphery.  Figure~\ref{Drug} visualizes a drugs-crime event   in Manchester, UK, where the events are connected if they are located within a distance of (left) 1Km  and (right) 5Km of each other. The high risk events form clusters in   the top-level visualizations.  
  
\begin{figure*}
\centering 
{\includegraphics[width=\columnwidth, height = 6cm]{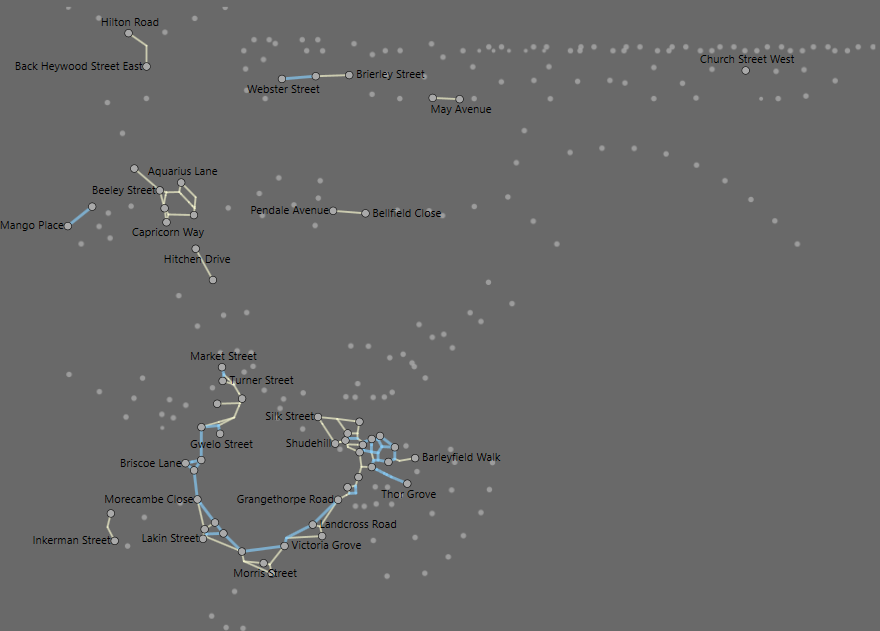}}
\hfill
{\includegraphics[width=\columnwidth, height = 6cm]{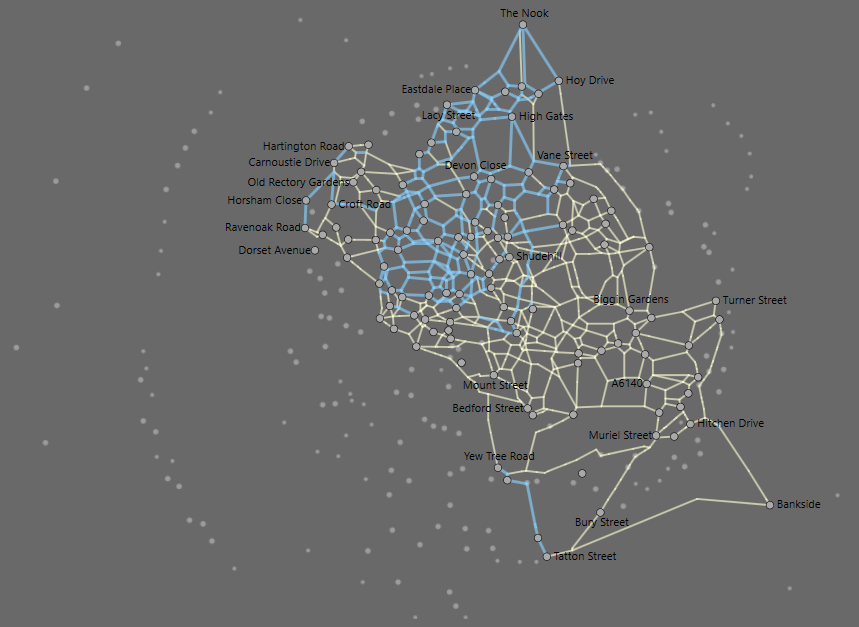}}
\caption{A visualization of drugs-crime event   in Manchester, UK, with approximately (left) 0.5K edges, and (right) 8K edges. The rails in light blue and white illustrate the first and second zoom levels, respectively.}
\label{Drug}
\end{figure*}

Figure~\ref{fig:z3incl} shows an experiment  with the graph of include dependencies of the C++ sources of\footnote{\href{https://github.com/Z3Prover/z3}{\url{https://github.com/Z3Prover/z3}}}.  The developer was interested in how his files are used by the rest of the system. He clicked node `lar\_solver.h' representing an important header file of his files and created the neighborhood in red color (the upper drawing). Then he noticed that file `theory\_lra.cpp' includes `lar\_solver.h' and clicked the former, creating the blue neighborhood (in the lower drawing). Then he noticed two files, marked by the black oval, that were included into `theory\_lra.cpp' by mistake. 

In another experiment, a user was analyzing collaboration between Chinese and Russian composers in 20-th and 21-st centuries. By highlighting the neighborhoods of Chinese composers of 20-th century, the user saw that there were no connections between the composers of these two countries in this period. In 21-st century the only relation of such kind that he found was between Tan Dun\footnote{  \href{https://en.wikipedia.org/wiki/Tan\_Dun}{\url{https://en.wikipedia.org/wiki/Tan\_Dun}}} and Sofia Gubaidulina\footnote{\href{https://en.wikipedia.org/wiki/Sofia\_Gubaidulina}{\url{https://en.wikipedia.org/wiki/Sofia\_Gubaidulina}}}. 
\begin{figure*}[pt]
\centering 
{\includegraphics[width=\columnwidth, height = 2.5cm]{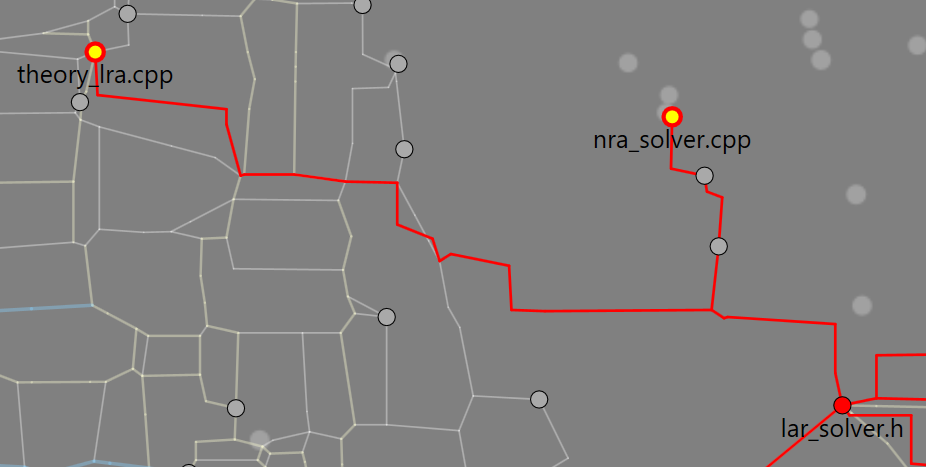}}
\hfill
{\includegraphics[width=\columnwidth, height = 2.5cm]{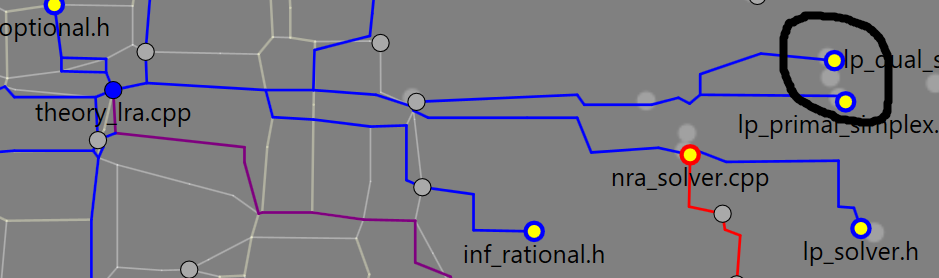}}
\caption{Highlighting the neighborhood in a unique color helps to understand relations.}
\label{fig:z3incl}
\end{figure*} 
More details can be seen in a video\footnote{\href{https://www.youtube.com/watch?v=qCUP20dQqBo&feature=youtu.be}{\url{https://www.youtube.com/watch?v=qCUP20dQqBo&feature=youtu.be}}}.
 % GraphMaps is a part of \href{https://github.com/Microsoft/automatic-graph-layout}{MSAGL}~\footnote{\href{https://github.com/Microsoft/automatic-graph-layout}{https://github.com/Microsoft/automatic-graph-layout}}. The experiments were conducted by running TestGraphmaps.exe, which was built under Release/x64 configuration.

\section{\uppercase{Limitations}}
Adjacency relations and node degrees are readily visible in small size traditional node-link diagrams, GraphMaps can process large networks, but it looses those two aspects at an expense of  avoiding clutter. The users need to select nodes to explore the adjacencies and node degrees. Currently, we use colors to disambiguate node selections, which limits us to the selection of a small number of nodes avoiding ambiguity. GraphMaps is sensitive to node quota or the maximum number of nodes per tile. Selecting a large node quota may increase the interaction latency during level transitions. On the other hand, selecting a small node quota may select few nodes on the top-level, which may fail to give an overview of the graph structure. An  appropriate choice of the node quota based on the graph size and node layout is yet to discover. 
For simplicity, we used polygonal chains to represent the edges, different colors for multiple node selections, and   color transparency to avoid ambiguity. It would be interesting to find ways of improving the visual appeal of a GraphMaps visualization, e.g., using splines for edges, enabling tooltip texts for showing quick information and so on.

\section{\uppercase{Conclusion}}

We described our algorithm to construct GraphMaps Visualizations using competition mesh. Recall that the edge stretch factor of the competition mesh we created is at most $(2+\sqrt{2})$. A natural open question is to establish  tight bound on the edge stretch factor of the competition mesh. It would also be interesting to find bounded degree spanners (possibly with Steiner points) that are monotone and have low stretch factor. We refer the reader to~\cite{BoseS13,DBLP:journals/jgaa/DehkordiFG15,DBLP:conf/compgeom/FelsnerIKKMS16} for more details on such  geometric mesh and spanners. We leave it as a future work to examine how the quality of a GraphMaps system may vary depending on the choice of geometric mesh.

\begin{figure}[h]
  \centering 
  {\includegraphics[width=.45\textwidth]{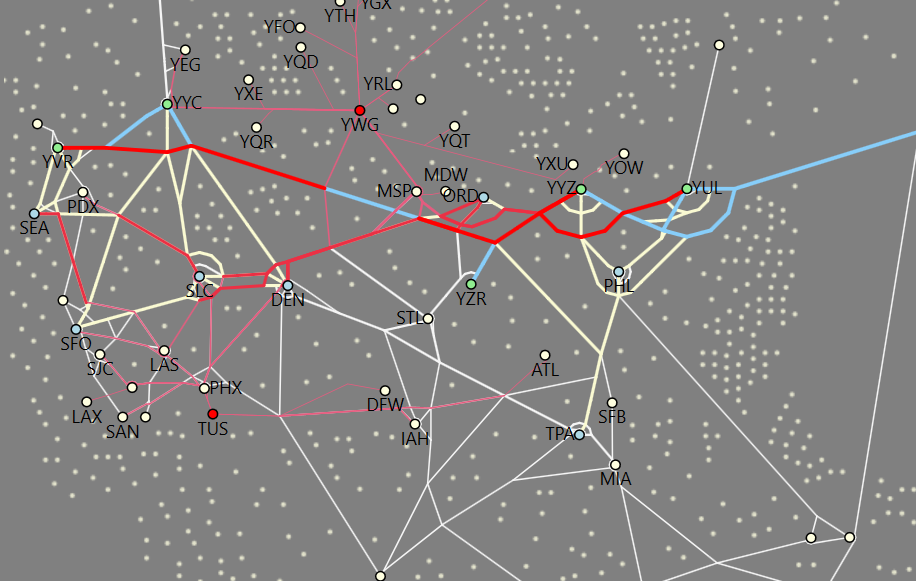}}  
  \caption{A visualization of the flight network dataset (https://openflights.org/data.html)  using previous GraphMaps~\cite{Nachmanson15}.}
   \label{abs2} 
\end{figure}

The previous GraphMaps~\cite{Nachmanson15} uses an incremental mesh generation, which does not require path simplification. Since the construction of the upper levels does not take the lower level nodes into account, the top-level view is usually sparse, e.g., see Figure~\ref{abs2}. Our approach is powerful in the sense that any mesh can be transformed into a GraphMaps visualization. But the upper levels are the  simplification of the bottom level mesh, and thus the quality of the top-level depends on both the bottom level mesh and the simplification process. It will be interesting to further explore the pros and cons of both approaches. We believe that our results will inspire further research to enhance the appeal and usability of GraphMaps visualizations.

%Another intriguing research direction would be to integrate GraphMaps with a  set of rich user interaction such as editing the network and moving the nodes.
%We observed that the edge stretch factor can be improved if we create the competition mesh with the left priority rule, i.e., If a ray $r$ reaches another orthogonal ray $r'$ on its way, then  $r$ stops only if $r'$ comes from the left direction.
% However, competition mesh created using left priorities may create $O(n^2)$ junctions. It would be very interesting to examine the trade-off between edge stretch factor and the number of junctions created in any orthogonal mesh.  

\section*{\uppercase{Acknowledgements}}

\noindent This work was initiated when the first author was a summer intern at Microsoft Research, Redmond, USA. His subsequent work was partially supported by NSERC.

%\vfill
\bibliographystyle{apalike}
{\small
\bibliography{bibs}}

\vfill
\end{document}